\documentclass{article}
\usepackage[utf8]{inputenc}
\usepackage[T1]{fontenc}
\usepackage{makeidx}  
\usepackage{amsmath}
\usepackage{amssymb}
\usepackage{graphicx}
\usepackage{nicefrac}
\usepackage{subfigure}
\usepackage{todonotes}	
\usepackage{amsthm}
\theoremstyle{plain}
\newtheorem{theorem}{Theorem}
\newtheorem{lemma}[theorem]{Lemma}
\newtheorem{corollary}[theorem]{Corollary}
\begin{document}
\title{Congestion Games with Mixed Objectives\thanks{This work was partially supported by the German Research Foundation (DFG) within the Collaborative Research Centre ``On-The-Fly Computing'' (SFB 901) and by the EU within FET project MULTIPLEX under contract no.\ 317532.\newline\newline The final publication is available at Springer via http://dx.doi.org/10.1007/978-3-319-48749-6\_47.}}
%
%
\author{Matthias Feldotto \and Lennart Leder \and Alexander Skopalik \\\\
	Heinz Nixdorf Institute \& Department of Computer Science\\
	 Paderborn University, Germany\\\\
\{feldi,lleder,skopalik\}@mail.upb.de
	}
%
%
%

\maketitle              

\begin{abstract} We study a new class of games which generalizes congestion games and its bottleneck variant. We introduce {\em congestion games with mixed objectives} to model network scenarios in which players seek to optimize for latency and bandwidths alike. We characterize the existence of pure Nash equilibria (PNE) and the convergence of improvement dynamics. For games that do not possess PNE we give bounds on the approximation ratio of approximate pure Nash equilibria.

\end{abstract}

\section{Introduction}

Resource allocation problems in large-scale scenarios such as networks often cannot be solved as a single optimization problem. The size of the problem, the distributed nature of information, or control preclude a centralized approach.
As a consequence, decisions are delegated to local actors or players. This gives rise to strategic behavior as these players often have economic interests. Game theory has studied the effect of such strategic interaction in various models of resource allocation and scheduling. One of the most prominent ones is the class of atomic congestion games~\cite{RO73}, in which players allocate sets of resources. The cost of a resource depends on the number of players allocating it. The cost of a player is the sum of the costs of her allocated resources. The appeal of this model stems not only from its applicability to prominent problems like scheduling, routing and load balancing, but also from desirable game theoretic properties. Congestion games always possess pure Nash equilibria and the natural improvement dynamics converge to a pure Nash equilibrium since these games are potential games. In fact, the class of congestion games is isomorphic to the class of potential games~\cite{MS96}, which shows their expressiveness. When modeling network routing with congestion games and most of its variants like weighted~\cite{FO05} or player-specific congestion games~\cite{MI96} one faces deficiency due to the nature of the players' cost functions. As a player's cost is determined by the sum of the resource costs, congestion games are not well suited to model effects like bandwidth allocation, as here the cost of a player is determined only by the bottleneck resource. Hence, Banner and Orda~\cite{BA07} introduced bottleneck congestion games where the cost of a player is the maximum cost of her chosen resources. Again, due to the nature of the cost functions, this class of games and most of its variants~\cite{HA09,HH14} only model the bottleneck effects and are unable to describe latency effects.
It is not difficult to envision scenarios in which both effects, latency and bandwidth, are relevant to decision makers - especially in today's IT infrastructures where we find techniques with shared resources, for example in the context of cloud computing or in software-defined networking. Many users with lots of different applications and therefore different objectives interact in one network and compete for the same resources. Consider, for example, on the one hand media streaming and on the other hand video gaming. In one application, bandwidth is the most important property, in the other it is latency.\\
Therefore, we study a game theoretic model in which players may have heterogeneous objectives. We introduce the model of \emph{congestion games with mixed objectives}. In this model, resources have two types of costs, latency cost and bottleneck cost, where the latter corresponds to the inverse of bandwidth. The players' costs may depend on both types of cost, where we allow different players to have different preferences regarding the two cost types.

{\bf \noindent Our Contribution.}
We show that pure Nash equilibria exist and can be computed in polynomial time in singleton games and in some matroid games. However, we show that the matroid property alone is not sufficient for the existence. Additionally, it is necessary that either the players are only interested in latency or bottleneck cost, or that the cost functions have a monotone dependence. For the latter case, we show convergence of best-response dynamics while the remaining cases are only weakly acyclic.
For matroid games that do not satisfy one of the additional properties, we show that pure equilibria might not exist and it is even {\sf NP}-hard to decide whether one exists.
To overcome these non-existence results, we consider approximate pure Nash equilibria. For several classes, we can show that there exist $\beta$-approximate pure Nash equilibria where $\beta$ depends on the size of the largest strategy.

{\bf \noindent Related Work.}
Milchtaich~\cite{MI96} studies the concept of player-specific congestion games and shows that in the singleton case these games always admit pure Nash equilibria. Ackermann et al.~\cite{AR09} generalize these results to matroid strategy spaces and show that the result also holds for weighted congestion games. 
Furthermore, they point out that in a natural sense the matroid property is maximal for the guaranteed existence of pure Nash equilibria in player-specific and weighted congestion games.
Moreover, Milchtaich~\cite{MI96} examines congestion games in which players are both weighted and have player-specific cost functions. By constructing a game with three players, he shows that these games do not necessarily possess pure Nash equilibria, even in the case of singleton strategies.\\
Mavronicalas et al.~\cite{MA07} study a special case of these games in which cost functions are not entirely player-specific. Instead, the player-specific resource costs are derived by combining the general resource cost function and a player-specific constant via a specified operation (e.\,g. addition or multiplication). They show that this restriction is sufficient to guarantee the existence of pure Nash equilibria in games with three players.
Dunkel and Schulz~\cite{DS06} show that the decision problem of whether a weighted network congestion game possesses a pure Nash equilibrium is NP-hard. The equivalent result is achieved for player-specific congestion games by Ackermann and Skopalik~\cite{AS08}.\\
Banner and Orda~\cite{BA07} study the applicability of game-theoretic concepts in network routing scenarios. In particular, they derive bounds on the price of anarchy in network bottleneck congestion games with restricted cost functions and show that a pure Nash equilibrium which is socially optimal always exists.
Cole et al.~\cite{CDR12} further investigate the non-atomic case, they especially consider the impacts of variable traffic rates.
In contrast, Harks et al.~\cite{HA09} concentrate on the atomic case and study the  lexicographical improvement property, which guarantees the existence of pure Nash equilibria through a potential function argument. They show that bottleneck congestion games fulfill this property and, hence, they are potential games.
Harks et al.~\cite{HH13} consider the complexity of computing pure Nash equilibria and strong equilibria in bottleneck congestion games. Moreover, they show this property in matroid bottleneck congestion games.  \\
Chien and Sinclair~\cite{CS07} study the convergence towards approximate pure Nash equilibria in symmetric congestion games. Skopalik and Vöcking~\cite{SK08} show inapproximability in asymmetric congestion games, which is complemented by approximation algorithms for linear and polynomial delay functions~\cite{CA11,FG14}, even for weighted games~\cite{CA12}. Hansknecht et al.~\cite{HK14} use the concept of approximate potential functions to examine of approximate pure Nash equilibria in weighted congestion games under different restrictions on the cost functions.

{\bf \noindent Preliminaries}
A {\em congestion game with mixed objectives} is defined by a tuple $\Gamma= (N,R,\left(\Sigma_i\right)_{i\in N},\left(\alpha_i\right)_{i\in N},\left(\ell_r\right)_{r\in R}, \left(e_r\right)_{r\in R})$, where  $N=\left\{1,\dots, n\right\}$ denotes the set of players and $R$ denotes the set of resources. For each player $i$ let $\Sigma_i \subseteq 2^R$ denote the strategy space of player $i$ and $\alpha_i\in \left [0,1\right ]$ the preference value of player $i$. For each resource $r$ let $\ell_r : N \to \mathbb{R}$ denote the non-decreasing latency cost function associated to resource $r$, and let $e_r : N \to \mathbb{R}$ denote the non-decreasing bottleneck cost function associated to resource $r$. 
For a state $S=\left (S_1,\dots, S_n\right ) \in \Sigma_1 \times \ldots \times \Sigma_n$, we define for each resource $r\in R$ by $n_r(S)=\left|\left \{i\in N~|~r\in S_i\right \}\right|$ the congestion of $r$. The latency cost of $r$ in state $S$ is given by $\ell_r(S)=\ell_r(n_r(S))$, and the bottleneck cost by $e_r(S)=e_r(n_r(S))$. The total cost of player $i$ in state $S$ depends on $\alpha_i$ and is defined as $c_i(S)=\alpha_i\cdot \sum_{r\in S_i}\ell_r(S)+(1-\alpha_i)\cdot\max_{r\in S_i}e_r(S)$.

For a state $S=(S_1,...,S_i,..., S_n)$, we denote by $\left(S_{i}',S_{-i}\right)$ the state that is reached if player $i$ plays strategy $S_i'$ while all other strategies remain unchanged. 
A state $S=\left (S_1,\dots, S_n\right)$ is called a {\em pure Nash equilibrium (PNE)} if for all $i\in N$ and all $S'_i\in \Sigma_i$ it holds that $c_i(S)\leq  c_i(S'_i,S_{-i})$ and a {\em $\beta$-approximate pure Nash equilibrium} for a $\beta \ge 1$, if for all $i\in N$ and all $S'_i\in \Sigma_i$ it holds that $c_i(S)\leq \beta\cdot c_i(S'_i,S_{-i})$.

A {\em singleton congestion game with mixed objectives} is a congestion game with mixed objectives $\Gamma$ with the additional restriction that all strategies are singletons, i.\,e., for all $i\in N$ and all $S_i\in \Sigma_i$ we have that $|S_i|=1$.
A {\em matroid congestion game with mixed objectives} is a congestion game with mixed objectives in which the strategy spaces of all players form the bases of a matroid on the set of resources.
We say that the cost functions of a congestion game with mixed objectives have a {\em monotone dependence} if  there is a monotone non-decreasing function $f:\mathbb{R}\to \mathbb{R}$, such that $e_r(x)=f(\ell_r(x))$ for all $r\in R$. We call the players {\em $\alpha$-uniform} if there is an $\alpha \in [0,1]$ such that $\alpha_i=\alpha$ for all players $i\in N$. We say the players have {\em pure preferences} if $\alpha_i \in \{0,1\}$ for all players $i\in N$.

\section{Existence of Pure Nash Equilibria}
\label{existence}

Congestion games with mixed objectives are more expressive than standard or bottleneck congestion games. Consequently, the existence of pure Nash equilibria is 
guaranteed only for special cases. Unlike, e.\,g., player-specific congestion games, the matroid property is not sufficient for the existence of PNE. We show that we have the existence of PNE in singleton games or for matroid games with players that have pure preferences or cost functions that have a monotone dependence.

\begin{theorem}
	\label{Mat01}
	A congestion game with mixed objectives  $\Gamma$ contains a pure Nash equilibrium  if  $\Gamma$ is a
	\begin{enumerate}
		\item singleton congestion game, or
		\item matroid congestion game and the players have pure preferences, or
		\item matroid congestion game and the cost functions have a monotone dependence.
	\end{enumerate}
	A pure Nash equilibrium can be computed in polynomial time.
\end{theorem}

\begin{proof}
	
	We prove the theorem by reducing the existence problem of a pure Nash equilibrium in a congestion game with mixed objectives to the existence problem of a PNE in a congestion game with player-specific cost functions.
	The existence of PNE is guaranteed in singleton~\cite{MI96} and matroid~\cite{AR09} player-specific congestion games and polynomial time complexity immediately follows~\cite{AC08,AR09}. 
	
	We will utilize the following lemma which states that an optimal basis with respect to sum costs is also optimal w.\,r.\,t. maximum costs.

	\begin{lemma}
		\label{Mat-Lemma}
		Let $M$ be a matroid, and let $B=\{b_1,\dots,b_m\}$ be a basis of $M$ which minimizes the sum of the element costs. Then for any other basis $B'=\{b_1',\dots,b_m'\}$ it holds that $\max_{1\leq i \leq m}b_i \leq \max_{1\leq i \leq m}b_i'$.
	\end{lemma}
	\begin{proof}[Lemma]
		Let $B=\{b_1,\dots,b_m\}$ be an optimal basis, and assume by contradiction that there is a different basis $B'=\{b_1',\dots,b_m'\}$ with $b_m'<b_m$ (w.\,l.\,o.\,g. assume that $b_m$ and $b_m'$ are the most expensive resources in $B$ and $B'$, respectively). Since $B$ and $B'$ are both bases and $b_m\notin B'$, there is an element $b_i'\in B'$ such that $B''=B\setminus \{b_m\}\cup \{b_i'\}$ is a basis of $M$. By assumption we have $b_i'\leq b_m'<b_m$, which implies that $B''$ has a smaller total cost than $B$. Therefore, $B$ cannot be optimal, which gives a contradiction.
		\qed
	\end{proof}
	
	We now proceed to prove Theorem~\ref{Mat01} and consider the three different cases:
	\begin{enumerate}
		\item  The cost of player $i$ in a state $S$ with $S_i=\{r\}$ is
		$c_i(S)=\alpha_i\cdot \ell_r(S)+\left(1-\alpha_i\right)\cdot e_r(S)$.
		By defining the player-specific cost functions $c_r^i(x)= \alpha_i\cdot \ell_r(x)+\left(1-\alpha_i\right)\cdot e_r(x)$ for every $i\in N$ and $r\in R$, we obtain an equivalent singleton player-specific congestion game.
		
		\item Using Lemma~\ref{Mat-Lemma}, we can treat all the players who strive to minimize their bottleneck costs as if they were striving to minimize the sum of the bottleneck costs of their resources. Hence, we can  construct a player-specific congestion game in which the player-specific cost functions correspond to the latency functions for those players with preference value 1, and to the bottleneck cost functions for those players with preference value 0.
		
		\item A player $i$ who allocates the resources $\{r_1,\dots,r_k\}$, where w.\,l.\,o.\,g. $r_k$ is the most expensive one, in a state $S$, incurs a total cost of
		$c_i(S)=\alpha_i\cdot \sum_{j=1}^k \ell_{r_j}(S)+\left(1-\alpha_i\right) \cdot e_{r_k}(S)\, .$ 
		Due to Lemma \ref{Mat-Lemma}, we know that $\ell_{r_k}(S)$ is minimized if $\sum_{j=1}^k \ell_{r_j}(S)$ is minimized. The monotonicity of $f$ with $e_r(x)=f(\ell_r(x))$ implies that $e_{r_k}$ is also minimized. Observe that monotonicity also ensures that $r_k$ is the bottleneck resource. Hence, a PNE of a congestion game with cost functions $\ell_r$ for every $r\in R$ is a PNE of $\Gamma$.	\qed\end{enumerate}
	
\end{proof}

We show that the matroid property is not sufficient for the existence of PNE. Even for linear cost functions and uniform players, there are games without PNE. 

\begin{theorem} 
	\label{alphaequal}
	There is a matroid congestion game with mixed objectives $\Gamma$ with linear cost functions and $\alpha$-uniform players which does not possess a pure Nash equilibrium.
\end{theorem}

\begin{proof}
	
	We construct a two-player game with linear cost functions and $\alpha_i=0.5$ for both players.	
	The set of resources is $R=\left \{r_1,\dots,r_7\right\}$. The strategies for player $1$ are $\Sigma_1=\left\{\left\{r_i,r_j,r_k\right\}~|~i,j,k\in\{1,\dots,6\}\right\}$ and the strategies for player $2$ are $\Sigma_2=\left\{\left\{r_i,r_j\right\}~|~i,j\in\{4,\dots,7\}\right\}$.
	The latency and bottleneck cost functions for the first three resources are $\ell_{r_1}(x)=\ell_{r_2}(x)=\ell_{r_3}(x)=0$ and $e_{r_1}(x)=e_{r_2}(x)=e_{r_3}(x)=200\cdot x$, respectively. 
	For resource $r_4$ and $r_5$ the cost functions are $\ell_{r_4}(x)=\ell_{r_5}(x)=20\cdot x$ and $e_{r_4}(x)=e_{r_5}(x)=50\cdot x$. For resource $r_6$ the cost functions are $\ell_{r_6}(x)=8\cdot x$ and $e_{r_6}(x)=80\cdot x$. For resource $r_7$ the cost functions are $\ell_{r_7}(x)=0$ and $ e_{r_7}(x)=160\cdot x$.
	
	We note that for player 1 only the strategies $S_{1,1}:=\{r_1,r_2,r_3\}$ and $S_{1,2}:=\{r_4,r_5,r_6\}$ can be best-response strategies in any state, since $\{r_1,r_2,r_3\}$ strictly dominates all remaining strategies.
	Hence, with respect to the existence of pure Nash equilibria, we can restrict player 1 to these two strategies. In the analogous way, we can restrict player 2 to the strategies $S_{2,1}:=\{r_4,r_5\}$ and $S_{2,2}:=\{r_6,r_7\}$. This yields a game with only four states and, as we can easily verify, a best-response improvement step sequence starting from any of these states runs in cycles:
	$$\begin{pmatrix}100 && 45\\S_{1,1}&&S_{2,1}\end{pmatrix} \xrightarrow{1} \begin{pmatrix}94 && 90\\S_{1,2}&&S_{2,1}\end{pmatrix} \xrightarrow{2} \begin{pmatrix}108 && 88\\S_{1,2}&&S_{2,2}\end{pmatrix}$$
	$$\xrightarrow{1} \begin{pmatrix}100 && 84\\S_{1,1}&&S_{2,2}\end{pmatrix} \xrightarrow{2} \begin{pmatrix}100 && 45\\S_{1,1}&&S_{2,1}\end{pmatrix}$$	
	The numbers above the strategies give the costs of the respective player in the described state, and the numbers on the arrows indicate which player changes her strategy from one state to the next one.
	\qed
\end{proof}

Note that the nature of the existence proofs of Theorem~\ref{Mat01} implies that a PNE can be computed in polynomial time. However, if existence is not guaranteed, the decision problem whether a matroid game has a pure Nash equilibrium is {\sf NP}-hard:

\begin{theorem}
	\label{matroidcomp}
	It is {\sf NP}-hard to decide whether a matroid congestion game with mixed objectives  possesses at least one pure Nash equilibrium even if the players are $\alpha$-uniform.
\end{theorem}

\begin{proof}
	We reduce from Independent Set (IS), which is known to be {\sf NP}-complete \cite{GJ02}. Let the graph $G=(V,E)$ and $k\in \mathbb{N}$ be an instance of IS. We construct a matroid congestion game $\Gamma$ that has a pure Nash equilibrium if and only if $G$ has an independent set of size at least $k$.
	
	We begin by describing the structure of $\Gamma$. The game contains the following groups of players:
	\begin{itemize}
		\item For each node $v\in V$, there is one player who can allocate all edges adjacent to $v$, but has a profitable deviation to a special strategy if and only if a player of a neighboring node allocates one of the adjacent edges.
		\item There are two players who play a game that is equivalent to the game in the proof of Theorem~\ref{alphaequal} if an additional player allocates a certain resource $r_7$, but possesses a PNE otherwise.
		\item Finally, there is one connection player for whom it is profitable to allocate $r_7$ if and only if at least $n-k+1$ of the node players deviate from their ``edge strategy''.
	\end{itemize}
	
	Clearly, if we can achieve this dynamic, the existence of a PNE in $\Gamma$ is equivalent to the existence of an independent set in $G$. 
	Let $V=\{v_1,\dots,v_n\}$, and for every $v_i\in V$ we denote by $E_{v_i}=\{e\in E~|~v_i\in e\}$ the set of edges adjacent to $v_i$, and by $d(v_i)=|E_{v_i}|$ the degree of $v_i$ in $G$. Let $d=\max_{v\in V} d(v)$ be the maximum degree in $G$. We can assume that $d\geq 2$.
	
	We now give a formal definition of $\Gamma=(N,R,\left(\Sigma_i\right)_{i\in N},\left(\alpha_i\right)_{i\in N}, \left(\ell_r\right)_{r\in R}, \linebreak \left(e_r\right)_{r\in R})$.
	The set of players is $N=\left\{v_1,\dots,v_n,c,1,2\right\}$ and the set of resources is $R=\{r_e~|~e\in E\} \cup \{r_i^j \mid i \in \{1,\ldots,n\}, j \in \{1,\ldots,d(v_i)-1\}\}\cup \{r_c,r_1,\dots,r_7\}$.
	The strategies for the vertex players $v_i$ are all subsets of size $d(v_i)$ from a set consisting of the adjacent edge resources, some alternative resources, and resource $r_c$:  
	$\Sigma_{v_i}=\left\{\{X~|~ X\subseteq \left(\{r_e~|~e\in E_{v_i}\} \cup \{q_i^1,\dots,q_i^{d(v_i)-1},r_c\}\right)\text{ and } |X|=d(v_i)\right\}$. 
	
	Our 
	choice of cost functions will ensure that in an equilibrium every vertex player $v_i$ either allocates all  resources $r_e$ that belong to her adjacent edges $e \in E_{v_i}$ or the resources $q_i$ and $r_c$.
	The two strategies of the connection player are $\Sigma_c=\{\{r_c\},\{r_7\}\}$. Finally, there are the players $1$ and $2$ with strategies 
	$\Sigma_1=\left\{\left\{r_i,r_j,r_k\right\} \mid i,j,k\in\{1,\dots,6\}\right\}$ and$, \Sigma_2=\left\{\left\{r_i,r_j\right\}~|~i,j\in\{4,\dots,7\}\right\}$, respectively.
	
	The cost functions for the edge resources are
	$\ell_{r_e}(x)=1000 \cdot x$ and $e_{r_e}(x)=0$ for all $e\in E$, for the alternative resources
	$\ell_{q_i^j}(x)=0$ and $e_{q_i^j}(x)= 1000 \cdot d(v_i) + 1$ for all $1\le i \le n$ and $1 \le j \le d$, and for the connection resource
	$\ell_{r_c}(x)=0$ and $e_{r_c}(x)=0$ for $x \le n-k+1$, $e_{r_c}(x)=1000$ for $x >n-k +1$.
	The cost functions of the resources $r_1,\dots,r_7$ are $\ell_{r_1}(x)=\ell_{r_2}(x)=\ell_{r_3}(x)=0, e_{r_1}(x)=e_{r_2}(x)=e_{r_3}(x)=200\cdot x,
	\ell_{r_4}(x)=\ell_{r_5}(x)=20\cdot x,\ e_{r_4}(x)=e_{r_5}(x)=50\cdot x,   \ell_{r_6}(x)=8\cdot x,\ e_{r_6}(x)=80\cdot x,\ \ell_{r_7}(x)=0,\ e_{r_7}(x)=80\cdot x$.
	We choose the value $\alpha_i=0.5$ for all players.
	
	It remains to show that this game has a pure Nash equilibrium if and only if $G$ has an independent set of size $k$.
	If there is an independent set, we can construct an equilibrium as follows: Each node player that corresponds to a node in the independent set chooses the strategy that contains all her edge resources. Each remaining node player chooses a strategy that contains only her $q$-resources and the resource $r_c$. The connection player chooses resource $r_c$. Player $1$ chooses $\{r_1,r_2,r_3\}$ and player $2$ chooses $\{r_6,r_7\}$.  It is easy to verify that this is indeed a pure Nash equilibrium.
	
	If there is no independent set of size at least $k$, we argue that in an equilibrium there are more than $n-k$ of the node players on resource $r_c$. Observe that for a node player that allocates at least one of the $q$-resources it is the best response to allocate the remaining $q$-resources and the resource $r_c$ for no additional cost. Furthermore, if two node players allocate the same edge resource, their best response is to choose the $q$-resources and $r_c$.
	Hence, the best response of the connection player is $\{r_7\}$ and players $1$ and $2$ will play the subgame defined in Theorem~\ref{alphaequal} that does not have a pure Nash equilibrium. 
	\qed
\end{proof}

In Theorem~\ref{Mat01} we characterized restrictions on the preference values and cost functions that guarantee the existence of PNE in congestion games with mixed objectives, when combined with the matroid property of strategy spaces. The following theorem shows that the matroid property is necessary even if we impose the additional constraint that bottleneck and latency cost functions are identical and linear.

\begin{theorem}
	\label{Non-Mat}
	There exists a congestion game with mixed objectives 
	with linear cost functions $\ell_r=e_r$ for all resources $r\in R$, and either
	\begin{enumerate}
		\item pure preferences, or
		\item $\alpha$-uniform players 
	\end{enumerate} which does not possess a pure Nash equilibrium.
\end{theorem}

\begin{proof}	
	We show the correctness of the statement by constructing two games which fulfill the preconditions stated in the two cases of the theorem and which do not have a pure Nash equilibrium.
	\begin{enumerate}
		\item
		We define the game with two players with $\alpha_1=0$ and $\alpha_2=1$. The game has six resources $R=\{r_1,r_2,\dots,r_6\}$. Each player has two strategies, thus $\Sigma_1=\left\{\{r_1\},\{r_2,r_3,r_4,r_5\}\right\}$ and $\Sigma_2=\left\{\{r_2,r_3,r_4\},\{r_5,r_6\}\right\}$. The latency and bottleneck costs are given by $\ell_{r_1}(x)=6\cdot x$, $\ell_{r_2}(x)=\ell_{r_3}(x)=\ell_{r_4}(x)=2\cdot x$, $\ell_{r_5}(x)=4\cdot x$, $\ell_{r_6}(x)=3\cdot x$ with $e_r(x)=\ell_r(x)$ for all $r\in R$
		We utilize the fact that bottleneck players prefer to allocate many cheap resources, while players who are interested in latency are more willing to share a single expensive resource.

		Let $S_{1,1}$ and $S_{1,2}$ denote the strategies of player 1, and $S_{2,1}$ and $S_{2,2}$ the strategies of player 2. Then we have the following cycle of improvement steps which visits all four states:

		$$\begin{pmatrix}6 & 6\\S_{1,1}&S_{2,1}\end{pmatrix} \xrightarrow{1} \begin{pmatrix}4 & 12\\S_{1,2}&S_{2,1}\end{pmatrix} \xrightarrow{2} \begin{pmatrix}8 & 11\\S_{1,2}&S_{2,2}\end{pmatrix}$$
		$$\xrightarrow{1} \begin{pmatrix}6 & 7\\S_{1,1}&S_{2,2}\end{pmatrix} \xrightarrow{2} \begin{pmatrix}6 & 6\\S_{1,1}&S_{2,1}\end{pmatrix}$$
		
		The numbers above the strategies give the cost of the respective player in the associated state, and the numbers on the arrows indicate which player has to change her strategy in order to get to the next state. As we see, every change in strategy decreases the cost of the player performing it. 
		Hence, none of the four states is a pure Nash equilibrium.\medskip
		
		\item
		We prove the theorem for the example $\alpha=0.5$. However, it is easily generalizable to arbitrary values between 0 and 1. The idea is to construct a game with two players in which one player always allocates an expensive resource. In addition this, both players allocate two resources and share exactly one of these resources. In detail we have the resources $R=\{r_1,r_2,\dots,r_5\}$ and the strategy sets $\Sigma_1=\left\{\{r_1,r_2,r_4\},\{r_1,r_3,r_5\}\right\}$ and $\Sigma_2=\left\{\{r_2,r_5\},\{r_3,r_4\}\right\}$. The latency and bottleneck costs are given by $\ell_{r_1}(x)=32\cdot x$, $\ell_{r_2}(x)=\ell_{r_3}(x)=14\cdot x$ and  $\ell_{r_4}(x)=\ell_{r_5}(x)=12\cdot x+8$ with $e_r(x)=\ell_r(x)$ for all $r\in R$.
		Depending on which resource is shared, the players allocate either two resources with medium costs or one cheap and one expensive resource, where the sum of the two resource costs is slightly smaller in the latter case.
		
		The second player prefers the first alternative, since she has to pay an additional price for her most expensive resource. On the other hand, the first player always allocates an expensive resource, hence she incurs no additional costs when allocating a cheap and an expensive resource.
		
		The strategy spaces are constructed in such a way that in every state the players share exactly one of the resources $\{r_2,\dots,r_5\}$. Let $S_1$ denote a state in which $r_2$ or $r_3$ is shared, and $S_2$ a state in which $r_4$ or $r_5$ is shared. Then the players incur the following costs: $c_1(S_1)=56, c_2(S_1)=38, c_1(S_2)=55, c_2(S_2)=39$.
		
		As we see, player 1 prefers the state $S_2$, while player 2 prefers $S_1$. Since both players always have the possibility to deviate to the other state, there is no state in which none of the players can improve her costs, and hence $\Gamma$ possesses no pure Nash equilibrium.
	\end{enumerate}
	\qed
\end{proof}

\section{Convergence}

In this section we investigate in which games convergence of best-response improvement sequences to a pure Nash equilibrium can be guaranteed. 
Perhaps surprisingly, there are singleton games in which best-response improvement sequences may run in cycles. This is even true for games with pure preferences.

\begin{theorem}
	\label{singleton-circles}
	There are singleton congestion games with mixed objectives and pure preferences in which best-response improvement sequences may run in cycles.
\end{theorem}

\begin{proof}
	We prove the theorem by constructing a singleton game with three players and showing that there exists a cyclic best-response improvement sequence in certain states. The game consists of three resources $R=\{r_1,r_2,r_3\}$ and the players have the strategy sets $\Sigma_1=\left\{\{r_1\},\{r_2\}\right\}$, $\Sigma_2=\left\{\{r_1\},\{r_3\}\right\}$ and $\Sigma_3=\left\{\{r_2\},\{r_3\}\right\}$. Two players prefer the latency costs $\alpha_1=\alpha_2=1$, one the bottleneck costs $\alpha_3=0$. The latency and bottleneck costs are given by $\ell_{r_1}=(2,5)$, $\ell_{r_2}=(3,4)$, $\ell_{r_3}=(1,6)$ and $e_{r_2}=(1,4)$,  $e_{r_3}=(2,3)$. The first number in the cost functions gives the cost if the resource is used by one player, the second number gives the cost if two players use it. The bottleneck cost function of $r_1$ is irrelevant since it is only used by players 1 and 2.
	
	In this game, the following cycling best-response improvement sequence can occur (set braces omitted for better readability):
	
	$$\begin{pmatrix}5 & 5 & 1\\ r_1 & r_1 & r_2 \end{pmatrix} \xrightarrow{1} \begin{pmatrix}4 & 2 & 4\\ r_2 & r_1 & r_2 \end{pmatrix} \xrightarrow{2} \begin{pmatrix}4 & 1 & 4\\ r_2 & r_3 & r_2 \end{pmatrix} \xrightarrow{3} \begin{pmatrix}3 & 6 & 3\\ r_2 & r_3 & r_3 \end{pmatrix}$$
	
	$$\xrightarrow{1} \begin{pmatrix}2 & 6 & 3\\ r_1 & r_3 & r_3 \end{pmatrix} \xrightarrow{2} \begin{pmatrix}5 & 5 & 2\\ r_1 & r_1 & r_3 \end{pmatrix} \xrightarrow{3} \begin{pmatrix}5 & 5 & 1\\ r_1 & r_1 & r_2 \end{pmatrix}$$
	
	The numbers on the arrows indicate which player has to change her strategy in order to reach the next state. The variable $r_i$ denotes the resource which is used by the corresponding player and the number on top gives the cost value for this player. We can verify that each change in strategy is beneficial for the player performing it (since every player has only two strategies, every improving strategy is a best-response strategy).
	
	Hence, we have a cycle of six states that are visited during this best-response improvement sequence. The pure Nash equilibria $(r_1,r_3,r_2)$ and $(r_2,r_1,r_3)$ are never reached.  
	\qed
\end{proof}

Note that due to our reduction in the proof of Theorem~\ref{Mat01}, we know that there exists a sequence that leads to an equilibrium~\cite{AR09}.

\begin{corollary}
	\hfill
	\begin{enumerate}
		\item Singleton congestion games with mixed objectives are weakly acyclic.
		
		\item Matroid congestion games with mixed objectives that have pure preferences are weakly acyclic.
		
	\end{enumerate}
\end{corollary}

We now turn to matroid games with a monotone dependence and show that they converge quickly to a PNE if the players perform lazy best-response moves. That is, if players perform a best response move, they choose the best-response strategy that has as many resources in common with the previous strategy as possible.
\begin{theorem}
	\label{convergence01}
	Let $\Gamma$ be a matroid congestion game with mixed objectives with cost functions that have a monotone dependence. Then any sequence of lazy best-response improvement steps starting from an arbitrary state in $\Gamma$ converges to a pure Nash equilibrium after a polynomial number of steps. 
\end{theorem}

\begin{proof}	
	The proof idea is based on the proof by Ackermann et al.~\cite{AC08} which shows that matroid congestion games guarantee polynomial convergence to PNE.\\
	We consider an increasingly ordered enumeration of all latency values that can occur in $\Gamma$ (an enumeration of the values $\{\ell_r(x)~|~r\in R,x\in N\}$). Let $\ell'_r(x)$ denote the position of the respective cost value in the enumeration.
	
	We define the following variant of Rosenthal's potential function:
	$\Phi(S)=\sum_{r\in R}\sum_{i=1}^{n_r(S)}\ell'_r(i)$.
	If $n=|N|$ denotes the number of players, and $m=|R|$ the number of resources in $\Gamma$, then there are at most $n\cdot m$ different cost values in the game. Hence, the value of $\Phi$ is upper bounded by $n^2\cdot m^2$. Thus, it suffices to show that every lazy best-response improvement step decreases the value of $\Phi$ by at least 1.
	
	If a player replaces a single resource $r$ by another resource $r'$ in a lazy best-response with $\alpha_i \cdot \ell_{r'}(S') +(1-\alpha_i) \cdot e_{r'}(S') <\alpha_i \cdot \ell_{r}(S) +(1-\alpha_i) \cdot e_{r}(S)$, then due to  the monotone dependence, we have $\ell_{r'}(S') < \ell_r(S)$. Hence $\ell_{r'}(S')$ must occur before $\ell_r(S)$ in the increasingly ordered enumeration of the cost values and we have $\ell'_{r'}(S') < \ell'_r(S)$.
	Thus, every sequence of lazy best-response improvement steps in $\Gamma$ terminates after a polynomial number of steps. 
	\qed
\end{proof}

We remark that the only reason to restrict the players to \emph{lazy} instead of arbitrary best-response strategies is that the players may have a preference value of exactly 0. If a player's cost is determined solely by her most expensive resource, she might be playing a best-response strategy by replacing her most expensive resource by a cheaper one and additionally replace another resource by a more expensive one. This additional exchange does not necessarily increase her costs, but it could lead to an increase in the value of the potential function. However, if all players have preference values different from 0, the theorem holds for arbitrary best-response improvement steps.

\section{Approximate Pure Nash Equilibria}
\label{approximate}

As PNE do not exist in general, we study the existence of  approximate equilibria. However, in general we cannot achieve an approximation factor better than $3$.

\begin{theorem}
	\label{linear-3}
	There is a congestion game with mixed objectives, in which all cost functions are linear, that does not contain a $\beta$-approximate pure Nash equilibrium for any $\beta<3$.
\end{theorem}

\begin{proof}
	We show the theorem by constructing a game with 10 players, in which in every state there is at least one player who can improve her costs by a factor of 3. A formal definition of the game is given by:
	
	\begin{center}
		$\Gamma=(N,R,\left(\Sigma_i\right)_{i\in N},\left(\alpha_i\right)_{i\in N},
		\left(\ell_r\right)_{r\in R},\left(e_r\right)_{r\in R})$ with
	\end{center}
	$N=\{1,\dots,10\}$, $R=\left\{r^i_1,r^i_2~|~i\in N\setminus\{5,6\}\right\}\cup$
	\\$\ \left\{r^{i,j}_{1,1},r^{i,j}_{1,2},r^{i,j}_{2,1},r^{i,j}_{2,2}~|~i<j \text{ and } \left(i,j\in \{1,2,3,4\}\text{ or } i,j\in \{7,8,9,10\}\right)\right\}$,
	\begin{align*}
		\Sigma_i= & \left\{\{r^i_1\} \cup \{r^{j,k}_{1,1},r^{j,k}_{2,1}~|~j=i\text{ or } k=i\},\right. \\ & \left.\{r^i_2\} \cup \{r^{j,k}_{1,2},r^{j,k}_{2,2}~|~j=i\text{ or } k=i\}\right\} \text{ for all }i\in N\setminus \{5,6\}\\
		\Sigma_{5}= & \left\{\{r^1_1,\dots,r^4_1,r^{7}_1,\dots,r^{10}_1, r^{i,j}_{1,2}, r^{k,l}_{1,2}~|~i,j\in \{1,\dots,4\},\ k,l\in \{7,\dots,10\}\right\}\\\cup & \left\{\{r^1_2,\dots,r^4_2,r^{7}_2,\dots,r^{10}_2, r^{i,j}_{1,1}, r^{k,l}_{1,1}~|~i,j\in \{1,\dots,4\},\ k,l\in \{7,\dots,10\}\right\}
		\\
		\Sigma_{6}= & \left\{\{r^1_1,\dots,r^4_1,r^{7}_2,\dots,r^{10}_2, r^{i,j}_{2,2}, r^{k,l}_{2,1}~|~i,j\in \{1,\dots,4\},\ k,l\in \{7,\dots,10\}\right\}\\\cup & \left\{\{r^1_2,\dots,r^4_2,r^{7}_1,\dots,r^{10}_1, r^{i,j}_{2,1}, r^{k,l}_{2,2}~|~i,j\in \{1,\dots,4\},\ k,l\in \{7,\dots,10\}\right\}
	\end{align*}
	$\alpha_i=1$ for all $i\in N\setminus\{5,6\}$ and $\alpha_i=0$ for $i\in \{5,6\}$,\\ 
	$\ell_r(x)=x$  and $e_r(x)=0$ for all $r\in\left\{r^i_j~|~i\in N\setminus \{5,6\},j\in\{1,2\}\right\}$,\\
	$\ell_r(x)=0$  and $e_r(x)=x$ for all $r\in\left\{r^{i,j}_{k,l}~|~i,j\in N\setminus \{5,6\},\ k,l\in\{1,2\}\right\}$
	
	The game contains three different groups of players and two different groups of resources. Players 1 to 4 and 7 to 10 each have two personal resources $r^i_1$ and $r^i_2$ among which they have to choose. These are the only resources on which they incur costs.
	
	In addition to these resources, there are two times two resources corresponding to each set consisting of two players from the same group (either 1 to 4 or 7 to 10). The two resources represent the two strategies which are available to these players, and exist distinctly for both players 5 and 6. If player $i$ plays her first strategy, then she also allocates all resources that correspond to sets in which $i$ is contained and represent the first strategy.
	
	In every state, the players 5 and 6 allocate one of the personal resources $r^i_j$ for each player $i\in \{1,\dots,4\}$ and $i\in \{7,\dots, 10\}$, where the $j$ is the same for all players among a group. Additionally, for both groups they have to allocate a resource that corresponds to one pair of players and represents the $j$ that they are not using (e.\,g., if player 5 allocates resource $r^i_1$ for all players $i\in\{1,\dots,4\}$, she must also allocate a resource that corresponds to a pair of players from $\{1,\dots,4\}$ and represents strategy 2, and an analogous resource for the players in $\{7,\dots,10\}$).
	
	These two additional resources are the ones on which players 5 and 6 incur costs. Since they are interested in their bottleneck costs, their costs are equal to the more expensive of the two.
	
	We can analyze this game by considering an arbitrary state. The strategy sets of players 5 and 6 are constructed in such a way that in every state they allocate the same personal resources of all players in $\{1,\dots,4\}$ or $\{7,\dots,10\}$.
	
	W.\,l.\,o.\,g. assume that they both allocate the resources $r^1_1,\dots,r^4_1$. This implies that the players 1 to 4 have a cost of $3$ if they use the first resource, and a cost of $1$ if they allocate their second resource. Hence, this state cannot be a $\beta$-approximate PNE with $\beta< 3$ if any of them play their first strategy.
	
	If, on the other hand, all these players play their second strategy, players 5 and 6 incur a cost of 3 on their resources $r^{i,j}_{1,2}$ and $r^{i,j}_{2,2}$, respectively, independent of which $i$ and $j$ they are using, since all players in $\{1,\dots,4\}$ allocate the resources corresponding to strategy 2. Both player 5 and 6 could decrease the cost of the resource corresponding to the first group to 1 by switching the strategy. However, we have to take the other group into account as well.
	
	If player 5 switches her strategy, she allocates the resource $r^{k,l}_{1,1}$ for freely selectable $k$ and $l$ in $\{7,\dots, 10\}$, which yields a cost of 1 if at least two players from $\{7,\dots,10\}$ play the second strategy. Analogously, player 6 incurs a cost of 1 on the resource corresponding to the second group if at least two players from this group play their first strategy. One of the two cases must hold, which implies that either player 5 or player 6 can improve her cost from 3 to 1. Hence, in every state there is at least one player who can improve her cost by a factor of 3, and there exists no $\beta$-approximate pure Nash equilibrium for any $\beta<3$. 
	\qed
\end{proof}

On the positive side we can show small approximation factors for small strategy sets. Besides matroid games, we can show approximation factors which are independent of the structure of the strategy sets, but depend on either $\alpha$-uniform players or on equal cost functions:

\begin{theorem}
	Let $\Gamma$ be a congestion game with mixed objectives. Let $d=\linebreak\max_{i\in N, S_i\in \Sigma_i}|S_i|$ be the maximal number of resources a player can allocate.
	\label{approximatetheorem}
	\begin{enumerate}
		\item If $\Gamma$ is a matroid congestion game, then $\Gamma$ contains a $d$-approximate pure Nash equilibrium.
		\item If the players are $\alpha$-uniform, then $\Gamma$ contains a $d$-approximate pure Nash equilibrium.
		\item If $e_r=\ell_r$ for all resources $r\in R$, then $\Gamma$ contains a $\sqrt{d}$-approximate pure Nash equilibrium.
		\item If the players are $\alpha$-uniform and $e_r=\ell_r$ for all resources $r\in R$, then $\Gamma$ contains a $\beta$-approximate pure Nash equilibrium for $\beta=\frac{d}{\alpha\cdot(d-1)+1}$.
	\end{enumerate} 
\end{theorem}

\begin{proof} 
	\begin{enumerate}
		\item
		The proof relies on the fact that PNE always exist in player-specific matroid congestion games~\cite{AR09}.
		We define a player-specific congestion game $\Gamma'$ with the following cost functions: $c^i_r(S)=\alpha_i\cdot \ell_r(S)+(1-\alpha_i)\cdot e_r(S)$.
		
		We show that every PNE in $\Gamma'$ is a $d$-approximate pure Nash equilibrium in $\Gamma$. Since PNE always exist in matroid player-specific congestion games, the claim follows.
		We denote by $c_i(S)$ the costs of player $i$ in state $S$ in $\Gamma$, and by $c_i^p(S)$ the costs in $\Gamma'$.		
		Let $S_i$ be a best-response strategy w.~r.~t. $S_{-i}$ in $\Gamma'$. Then we get for all strategies $S_i'\in \Sigma_i$:
		\begin{align*}
			c_i(S_i',S_{-i}) & =\alpha_i\cdot \sum_{r\in S_i'} \ell_r(S_i',S_{-i})+(1-\alpha_i)\cdot \max_{r\in S_i'}\ e_r(S_i',S_{-i})\\
			& \geq \alpha_i\cdot \sum_{r\in S_i'} \ell_r(S_i',S_{-i})+(1-\alpha_i)\cdot\frac{1}{d}\cdot \sum_{r\in S_i'}\ e_r(S_i',S_{-i})\\
			& \geq \frac{1}{d}\cdot c_i^p(S_i',S_{-i}) \geq \frac{1}{d}\cdot c_i^p(S_i,S_{-i})\geq \frac{1}{d}\cdot c_i(S_i,S_{-i})
		\end{align*}

		\item	
		We show that the  function 
		$\Phi(S)=\sum_{r\in R}\sum_{i=1}^{n_r(S)}\left(\alpha\cdot \ell_r(i)+(1-\alpha)\cdot e_r(i)\right)$
		is a $d$-approximate potential function, i.\,e., its value decreases in every $d$-improvement step.
		Consider a state $S$ and a player $i$ who improves her costs by a factor of more than $d$ by deviating to the strategy $S_i'$:

		\begingroup \thickmuskip=0mu \medmuskip=0mu \thinmuskip=0mu
		\begin{align*}
			&\Phi(S'_i,S_{-i})-\Phi(S)\\
			&=\sum_{r\in R}\sum_{i=1}^{n_r(S')}\left(\alpha\cdot \ell_r(i)+(1-\alpha)\cdot e_r(i)\right) -\sum_{r\in R}\sum_{i=1}^{n_r(S)}\left(\alpha\cdot \ell_r(i)+(1-\alpha)\cdot e_r(i)\right)\\
			&\leq\alpha\cdot \sum_{r\in S_i'}\ell_r(S')+(1-\alpha)\cdot\max_{r\in S_i'}e_r(S')+(1-\alpha)\cdot \sum_{r\in S_i'}e_r(S')-(1-\alpha)\cdot\max_{r\in S_i'}e_r(S')\\
			&-\left(\alpha\cdot \sum_{r\in S_i}\ \ell_r(S)+(1-\alpha)\cdot\max_{r\in S_i}e_r(S)\right)\\
			&\leq c_i(S')-c_i(S)+\sum_{r\in S_i'}(1-\alpha)\cdot e_r(S')-(1-\alpha)\cdot\max_{r\in S_i'}e_r(S')\\
			&\leq c_i(S')-c_i(S)+(d-1)\cdot c_i(S')<c_i(S')-d\cdot c_i(S')+(d-1)\cdot c_i(S')=0
		\end{align*}
		\endgroup

		\item
		We show that 
		$\Phi(S)=\sum_{r\in R}\sum_{i=1}^{n_r(S)}\ell_r(i)^2$
		is a  $\sqrt{d}$-approximate potential function.
		Consider a state $S$ that is minimizing $\Phi$ and player $i$ who deviates to the strategy $S_i'$. Note that 
		$\Phi(S)-\Phi(S_i',S{-i})=\sum_{r\in S}\ell_r(S)^2-\sum_{r\in S'}\ell_r(S')^2.$
		Hence, $\sum_{r\in S}\ell_r(S)^2\leq\sum_{r\in S'}\ell_r(S')^2$.
		
		\begin{align*}
			c_i(S_i',S{-i})&=\sum_{r\in S_i'}\alpha_i\cdot \ell_r(S')+(1-\alpha_i)\max_{r\in S_i'}\ell_r(S')\\
			&\geq \alpha_i\cdot\left(\sum_{r\in S_i'} \ell_r(S')^2\right)^\frac{1}{2}+(1-\alpha_i)\cdot \left(\max_{r\in S_i'}\ell_r(S')^2\right)^\frac{1}{2}\\
			&\geq \alpha_i\cdot\left(\frac{1}{\sqrt{d}}\cdot \sum_{r\in S_i}\ell_r(S)\right)+(1-\alpha_i)\cdot \left(\frac{1}{d}\cdot \left(\sum_{r\in S_i}\ell_r(S)^2\right)\right)^\frac{1}{2}\\
			&\geq \frac{\alpha_i}{\sqrt{d}}\cdot \sum_{r\in S_i}\ell_r(S)+\frac{1-\alpha_i}{\sqrt{d}}\cdot \max_{r\in S_i}\ell_r(S)= \frac{1}{\sqrt{d}}\cdot c_i(S).
		\end{align*}

		\item 
		We argue that
		$\Phi(S)=
		\sum_{r\in R}\sum_{i=1}^{n_r(S)}\ell_r(i)$ is an approximate potential function.
		If the latency and bottleneck cost functions are identical for each resource, we use the fact that the cost of the bottleneck resource is at least as high as the average latency cost, i.\,e.,
		$\max_{r\in S_i}e_r(S)\geq \frac{1}{|S_i|}\sum_{r\in S_i}\ell_r(S)\geq \frac{1}{d}\sum_{r\in S_i}\ell_r(S),$
		which implies
		$c_i(S)=\alpha\cdot \sum_{r\in S_i}\ell_r(S)+(1-\alpha)\cdot \max_{r\in S_i} \ell_r(S) \linebreak \geq \left(\alpha+\frac{1-\alpha}{d}\right)\sum_{r\in S_i}\ell_r(S).$

		Consider a state $S$ and a player $i$ who improves her costs by a factor of more than $\beta=\frac{d}{\alpha\cdot(d-1)+1}$ by deviating to the strategy $S_i'$.

		\begingroup \thickmuskip=0mu \medmuskip=0mu \thinmuskip=0mu
		\begin{align*}
			\Phi(S_i',S_{-i})-\Phi(S)&\leq c_i(S')-c_i(S)+\sum_{r\in S_i'}(1-\alpha)\cdot \ell_r(S')-(1-\alpha)\cdot\max_{r\in S_i'}\ell_r(S')\\
			&\leq c_i(S')-c_i(S)+(1-\alpha)\cdot\left(1-\frac{1}{d}\right)\cdot\sum_{r\in S_i'} \ell_r(S')\\
			&\leq c_i(S')-c_i(S)+(1-\alpha)\cdot\left(1-\frac{1}{d}\right)\cdot\frac{c_i(S')}{\alpha+\frac{1-\alpha}{d}}\\
			&=\left(1+\frac{(1-\alpha)\cdot \frac{d-1}{d}}{\alpha+\frac{1-\alpha}{d}}\right)c_i(S')-c_i(S)\\
			&=\left(1+\frac{d-1}{\alpha\cdot(d-1)+1}-1+\frac{1}{\alpha\cdot(d-1)+1}\right)c_i(S')-c_i(S)\\
			&=\frac{d}{\alpha\cdot(d-1)+1}c_i(S')-c_i(S)<0
			&\,\qed
		\end{align*}
		\endgroup
	\end{enumerate}
	
\end{proof}

We remark that the $\beta$ given in the fourth case of Theorem~\ref{approximatetheorem} is bounded from above by $\frac{1}{\alpha}$. However, if $\alpha$ is close to $0$, the bound of $\sqrt{d}$ derived for general games with $\ell_r=e_r$ for all resources, but without restrictions on the preference values, may give a better approximation guarantee.

\section{Conclusions}
We studied a new class of games in which players seek to minimize the sum of latency costs, the maximum of bottleneck costs, or a combination thereof. As a promising avenue for future work it would be interesting to consider other types of cost aggregation. This would be useful in scenarios with heterogeneous players with different interests in which the resources represent not only links in a network but also servers, routers, or any network functions in general.

\bibliographystyle{splncs03}

\begin{thebibliography}{10}
	\providecommand{\url}[1]{\texttt{#1}}
	\providecommand{\urlprefix}{URL }
	
	\bibitem{AR09}
	Ackermann, H., Röglin, H., Vöcking, B.: Pure nash equilibria in
	player-specific and weighted congestion games. Theoretical Computer Science
	410(17),  1552--1563 (2009)
	
	\bibitem{AC08}
	Ackermann, H., R\"{o}glin, H., V\"{o}cking, B.: {On the Impact of Combinatorial
		Structure on Congestion Games}. J. ACM  55(6),  25:1--25:22 (2008)
	
	\bibitem{AS08}
	Ackermann, H., Skopalik, A.: {Complexity of Pure Nash Equilibria in
		Player-Specific Network Congestion Games}. Internet Mathematics  5(4),
	323--342 (2008)
	
	\bibitem{BA07}
	Banner, R., Orda, A.: {Bottleneck Routing Games in Communication Networks}.
	IEEE Journal on Selected Areas in Communications  25(6),  1173--1179 (2007)
	
	\bibitem{CA11}
	Caragiannis, I., Fanelli, A., Gravin, N., Skopalik, A.: {Efficient Computation
		of Approximate Pure Nash Equilibria in Congestion Games}. In: IEEE 52nd
	Annual Symposium on Foundations of Computer Science (FOCS). pp. 532--541
	(2011)
	
	\bibitem{CA12}
	Caragiannis, I., Fanelli, A., Gravin, N., Skopalik, A.: {Approximate Pure Nash
		Equilibria in Weighted Congestion Games: Existence, Efficient Computation,
		and Structure}. ACM Trans. Econ. Comput.  3(1),  2:1--2:32 (2015)
	
	\bibitem{CS07}
	Chien, S., Sinclair, A.: Convergence to approximate nash equilibria in
	congestion games. Games and Economic Behavior  71(2),  315--327 (2011)
	
	\bibitem{CDR12}
	Cole, R., Dodis, Y., Roughgarden, T.: {Bottleneck links, variable demand, and
		the tragedy of the commons}. Networks  60(3),  194--203 (2012)
	
	\bibitem{DS06}
	Dunkel, J., Schulz, A.S.: {On the Complexity of Pure-Strategy Nash Equilibria
		in Congestion and Local-Effect Games}. Mathematics of Operations Research
	33(4),  851--868 (2008)
	
	\bibitem{FG14}
	Feldotto, M., Gairing, M., Skopalik, A.: {Bounding the Potential Function in
		Congestion Games and Approximate Pure Nash Equilibria}. In: Liu, T., Qi, Q.,
	Ye, Y. (eds.) Web and Internet Economics - 10th International Conference,
	{WINE} 2014, Beijing, China, December 14-17, 2014. Proceedings. Lecture Notes
	in Computer Science, vol. 8877, pp. 30--43. Springer (2014)
	
	\bibitem{FO05}
	Fotakis, D., Kontogiannis, S., Spirakis, P.: Selfish unsplittable flows.
	Theoretical Computer Science  348(2),  226--239 (2005)
	
	\bibitem{GJ02}
	Garey, M.R., Johnson, D.S.: {Computers and Intractability: {A} Guide to the
		Theory of NP-Completeness}. W. H. Freeman (1979)
	
	\bibitem{HK14}
	Hansknecht, C., Klimm, M., Skopalik, A.: {Approximate Pure Nash Equilibria in
		Weighted Congestion Games}. In: Jansen, K., Rolim, J.D.P., Devanur, N.R.,
	Moore, C. (eds.) Approximation, Randomization, and Combinatorial
	Optimization. Algorithms and Techniques (APPROX/RANDOM 2014). Leibniz
	International Proceedings in Informatics (LIPIcs), vol.~28, pp. 242--257.
	Schloss Dagstuhl--Leibniz-Zentrum fuer Informatik, Dagstuhl, Germany (2014)
	
	\bibitem{HH14}
	Harks, T., Hoefer, M., Schewior, K., Skopalik, A.: {Routing Games With
		Progressive Filling}. IEEE/ACM Transactions on Networking  24(4),  2553--2562
	(2016)
	
	\bibitem{HH13}
	Harks, T., Hoefer, M., Klimm, M., Skopalik, A.: Computing pure nash and strong
	equilibria in bottleneck congestion games. Mathematical Programming  141(1),
	193--215 (2013)
	
	\bibitem{HA09}
	Harks, T., Klimm, M., M{\"{o}}hring, R.H.: {Strong Nash Equilibria in Games
		with the Lexicographical Improvement Property}. In: Leonardi, S. (ed.)
	Internet and Network Economics, 5th International Workshop, {WINE} 2009,
	Rome, Italy, December 14-18, 2009. Proceedings. Lecture Notes in Computer
	Science, vol. 5929, pp. 463--470. Springer (2009)
	
	\bibitem{MA07}
	Mavronicolas, M., Milchtaich, I., Monien, B., Tiemann, K.: {Congestion Games
		with Player-Specific Constants}. In: Kucera, L., Kucera, A. (eds.)
	Mathematical Foundations of Computer Science 2007, 32nd International
	Symposium, {MFCS} 2007, Cesk{\'{y}} Krumlov, Czech Republic, August 26-31,
	2007, Proceedings. Lecture Notes in Computer Science, vol. 4708, pp.
	633--644. Springer (2007)
	
	\bibitem{MI96}
	Milchtaich, I.: Congestion games with player-specific payoff functions. Games
	and Economic Behavior  13(1),  111--124 (1996)
	
	\bibitem{MS96}
	Monderer, D., Shapley, L.S.: Potential games. Games and Economic Behavior
	14(1),  124--143 (1996)
	
	\bibitem{RO73}
	Rosenthal, R.W.: A class of games possessing pure-strategy nash equilibria.
	International Journal of Game Theory  2(1),  65--67 (1973)
	
	\bibitem{SK08}
	Skopalik, A., V\"{o}cking, B.: {Inapproximability of Pure Nash Equilibria}. In:
	Proceedings of the Fortieth Annual ACM Symposium on Theory of Computing. pp.
	355--364. STOC '08, ACM, New York, NY, USA (2008)
	
\end{thebibliography}

\end{document}